\documentclass[onecolumn,superscriptaddress,aps]{revtex4-2}
\usepackage{graphicx}
\usepackage{amsmath}
\usepackage{amsthm}
\usepackage{amssymb}
\usepackage{latexsym}
\usepackage{array}
\usepackage{booktabs}
\usepackage{hyperref}
\usepackage{amsfonts}
\usepackage{dsfont}
\usepackage{mathrsfs}
\usepackage{verbatim}
\usepackage{bbold}
\usepackage[normalem]{ulem}
\usepackage{upgreek}
\usepackage{makecell}
\usepackage{adjustbox}
\usepackage{algorithm}
\usepackage{algpseudocode}
\usepackage{color}
\usepackage{bm}
\usepackage{times}

\graphicspath{{figures/}}

\newcommand{\ket}[1]{\left|#1\right>}
\newcommand{\abs}[1]{\bigl|#1\bigr|}
\newcommand{\norm}[1]{\left\lVert#1\right\rVert}

\newcommand{\tr}[1]{\operatorname{Tr}\!\left[#1\right]}
\newcommand{\id}{\hat{\mathds{1}}}
\newcommand{\E}{\mathbb{E}}
\newcommand{\R}{\mathbb{R}}
\newcommand{\eps}{\varepsilon}
\newcommand{\pauli}{\mathsf{P}_1}
\newcommand{\calE}{\mathcal{E}}
\newcommand{\calP}{\mathcal{P}}
\newcommand{\hHtf}{\hat{H}^{\rm tf}}
\newcommand{\hUtf}{\hat{U}^{\rm tf}}
\newcommand{\be}{\bm{\beta}}
\newcommand{\sig}{\hat{\bm{\sigma}}}
\newcommand{\hJ}{\hat{J}}
\newcommand{\hH}{\hat{H}}
\newcommand{\hU}{\hat{U}}
\newcommand{\hP}{\hat{P}}
\newcommand{\hG}{\hat{G}}
\newcommand{\hg}{\hat{g}}
\newcommand{\hX}{\hat{X}}
\newcommand{\hY}{\hat{Y}}
\newcommand{\hZ}{\hat{Z}}
\newcommand{\hrho}{\hat{\rho}}
\newcommand{\hsigma}{\hat{\sigma}}
\newcommand{\hC}{\hat{C}}
\newcommand{\hVop}{\hat{V}}
\newcommand{\hW}{\hat{W}}
\newcommand{\hM}{\hat{M}}
\newcommand{\hDelta}{\hat{\Delta}}

\renewcommand\arraystretch{1.2}

\newtheorem{theorem}{Theorem}
\newtheorem{proposition}{Proposition}
\newtheorem{lemma}{Lemma}
\newtheorem{corollary}{Corollary}

\begin{document}

\title{Passive Pauli Toggling of Polarization Qubits in Optical Fibers: Error Bounds and QKD Performance}

\author{Bongjune~Kim}\email{bongjunekim@jejunu.ac.kr}
\affiliation{Department of Physics, Jeju National University, Jeju 63243, Republic of Korea}

\author{Jeongho~Bang}\email{jbang@yonsei.ac.kr}
\affiliation{Institute for Convergence Research and Education in Advanced Technology, Yonsei University, Seoul 03722, Republic of Korea}
\affiliation{Department of Quantum Information, Yonsei University, Incheon 21983, Republic of Korea}

\date{\today}

\begin{abstract}
Polarization qubits offer a direct route to fiber-based quantum communication, yet the fiber that carries them also scrambles their reference frame through uncontrolled birefringence. Active compensation commonly relies on monitoring and feedback, raising a natural question: can the link itself suppress coherent polarization drift before it reaches the receiver? We show that it can within a regime of sufficiently correlated unitary drift. Our central idea is to embed a fixed cyclic sequence of Pauli rotations along the fiber, turning propagation distance into a spatial toggling frame. Rather than estimating and inverting the unknown transformation, the sequence repeatedly reverses its leading action. We establish exact refocusing for constant generators compatible with a two-segment echo, show that a four-frame Pauli cell cancels the leading contribution of any traceless quasi-static generator, and bound the residual error for smoothly varying birefringence. We then connect these guarantees to operational BB84 quantities, including measured QBERs, the resulting secret fraction, and insertion loss. In simulations of a 50 km fiber with spatially correlated birefringence, a representative design reduces the mean QBER from $6.11\%$ to $0.224\%$ at a device spacing of $1.25$ km. With an assumed per-device transmission of $t=0.997$, this raises the asymptotic key rate per launched pulse by a factor of approximately $2.5$. For this parameter set, the best design in the tested scan is not the densest one: error suppression and optical loss create a finite operating window. These results provide a loss-aware design principle for passive polarization stabilization and suggest a low-overhead complement to active tracking in polarization-encoded QKD networks.
\end{abstract}

\maketitle

%-------------------------------------------------------------------------------------------------------------------------------------------------------------------------------------------------------------------------------------
\section{Introduction}
%-------------------------------------------------------------------------------------------------------------------------------------------------------------------------------------------------------------------------------------

Single-photon polarization qubits are among the most direct carriers for fiber-based quantum key distribution (QKD)~\cite{Bennett1984BB84,Gisin2002QuantumCryptography,Scarani2009PracticalQKD}. Their weakness is also elementary: even when attenuation is polarization independent, a birefringent fiber does more than reduce the photon flux; it applies an unknown unitary rotation to the polarization state. In a stable laboratory setup this rotation can be calibrated and compensated. In a deployed link, however, the transformation varies spatially because of manufacturing inhomogeneity and drifts over time with temperature, stress, and bending. Active polarization tracking is therefore a standard engineering tool, but it adds monitoring overhead and can be difficult to combine with fast, low-footprint, or passive nodes. The question addressed here is whether part of the coherent polarization drift can be suppressed by the link itself, before any measurement-based feedback is invoked.

In this work, we develop a mechanism that we call spatial Pauli toggling. Its basic logic is to make an unknown but slowly varying coherent rotation cancel against differently signed versions of itself. The mechanism may be viewed as a passive form of spatial ``bothering'' of the qubit. We insert a predetermined sequence of polarization rotations at fixed locations along the fiber, with each rotation chosen from the one-qubit Pauli set
\begin{eqnarray}
\pauli = \{\id,\hX,\hY,\hZ\},
\end{eqnarray}
up to physically irrelevant global phases. The photon encounters the same passive sequence on every transmission: no device estimates the instantaneous fiber transformation, and no controller updates the sequence in real time. The propagation coordinate $z$ then plays the role usually played by time, and the in-line rotations generate a spatial toggling frame in which slowly varying coherent errors acquire alternating signs.

With this mechanism in place, we ask when the resulting spatial cancellation is both mathematically controlled and operationally useful. We formulate fiber propagation as a path-ordered $SU(2)$ process and derive an exact toggling-frame factorization for embedded polarization rotations. This analysis yields two complementary cancellation statements: a two-segment echo exactly refocuses a constant generator that anticommutes with the inserted Pauli, whereas a cyclic four-frame cell visiting ${\id,\hX,\hY,\hZ}$ cancels the first moment of an arbitrary traceless quasi-static generator. For bounded segment generators, we further obtain an end-to-end norm bound that separates noncommuting higher-order contributions from genuine intra-cell variation; denoting the fixed total link length by $L$ and the spacing between adjacent control devices by $\ell$, its smooth-fiber corollary gives $O(L\ell)$ worst-case scaling as $\ell\rightarrow0$, with the generator-strength and Lipschitz bounds held fixed. We then carry these control guarantees to the protocol level. Pauli twirling preserves the two quantum bit error rates (QBERs), $Q_Z$ and $Q_X$, enabling bounds on the total Pauli error and the asymptotic secret fraction, while an explicit break-even condition determines whether the resulting improvement survives accumulated device loss. Direct simulations of a 50 km fiber with spatially correlated Ornstein--Uhlenbeck birefringence support this theoretical picture. At $\ell=1.25$ km, the mean QBER falls from $6.11\%$ to $0.224\%$, and the total Pauli error weight falls from approximately $9.1\times10^{-2}$ to $4.0\times10^{-3}$. For an assumed per-device transmission of $t=0.997$, the asymptotic key rate per launched pulse rises from $1.68\times10^{-2}$ to $4.23\times10^{-2}$ in this representative scan, an approximately $2.5$-fold increase. Further densification continues to improve polarization fidelity, but eventually lowers the net rate as insertion loss accumulates.

Several earlier studies pursued a broadly related direction: protecting flying polarization qubits through predetermined, spatially distributed control operations~\footnote{Early proposals translated temporal bang--bang control into spatially distributed polarization operations in optical fibers, while subsequent studies developed Pauli-frame, CPMG, XY-4, and Knill-type protocols, examined implementation errors and non-Markovian effects, and included an experimental demonstration in an optical cavity (see, for example, Refs.~\cite{WuLidar2004FiberNoise,MassarPopescu2007PMD,Damodarakurup2009BangBang,Lucamarini2011FlyingQubits,RoyBardhan2012FiberDD,RoyBardhan2013TailoredWavePlates,Gupta2015Knill,Barge2025PolarizationEntanglement}).}. What has remained missing, however, is a unified framework connecting the finite-spacing accuracy of such control directly to end-to-end quantum-communication performance after device loss is taken into account. We address this gap by elevating the established control principles into a quantitative link-level certification and design framework that jointly provides three elements: an explicit finite-spacing, end-to-end residual bound for generally noncommuting $SU(2)$ polarization drift under stated spatial-regularity assumptions; a channel-level conversion of that residual into bounds on basis-resolved BB84 QBERs and the asymptotic secret fraction; and a loss-aware break-even criterion that identifies the finite range of device spacings yielding a net communication advantage. By propagating a physical control bound through operational security metrics and device loss, our framework shifts the central question from whether spatial decoupling can suppress polarization drift to when, and by how much, it can improve an end-to-end quantum link. Therefore, the advance is a mathematically controlled, protocol-aware, and experimentally testable bridge between spatial control theory and loss-constrained BB84 system design.

Viewed in this way, the work provides a link-design framework for passive polarization stabilization at three complementary levels. At the control level, it identifies the spatial regularity required for effective open-loop symmetrization and quantifies the residual error accumulated over an entire fiber. At the protocol level, it expresses that residual in QBERs and secret fractions directly relevant to BB84, rather than only in the fidelity of a selected input state. At the engineering level, it replaces the idealized objective of ever denser decoupling with a finite operating window jointly determined by birefringence correlation, cell spacing, and insertion loss. Moreover, because the resulting decision rule can be evaluated from measured $Q_Z$, $Q_X$, and transmission, the framework can be applied to characterized links. Within the regime of sufficiently correlated coherent drift, the proposed scheme thus offers a testable route to passive link preconditioning and a principled foundation for future hybrid passive--active architectures.

%-------------------------------------------------------------------------------------------------------------------------------------------------------------------------------------------------------------------------------------
\section{Spatial model and embedded Pauli frames}
%-------------------------------------------------------------------------------------------------------------------------------------------------------------------------------------------------------------------------------------

We encode a qubit in the polarization basis $\{\ket{H},\ket{V}\}$. A lossless linear optical element that preserves the two polarization modes acts unitarily on their amplitudes~\cite{Agrawal2010Fiber,NielsenChuang2010}. Since a global phase has no operational effect on density operators, the relevant transformation may be represented by an element of $SU(2)$.

\begin{proposition}[Reduction of lossless polarization optics to $SU(2)$]
\label{prop:u2_to_su2}
A lossless linear transformation that mixes two orthogonal polarization modes induces a unitary Jones matrix $\hJ\in U(2)$ on the single-photon polarization subspace. Its action on density operators depends only on the equivalence class of $\hJ$ modulo a global phase, so it can be represented by an element of $SU(2)$.
\end{proposition}

\begin{proof}---Let $\hat{a}_H,\hat{a}_V$ be annihilation operators for the two polarization modes and write $\hat{a}'_i=\sum_j\hJ_{ij}\hat{a}_j$. Preservation of the canonical commutation relations, $[\hat{a}'_i,\hat{a}_j^{\prime\dagger}]=\delta_{ij}$, requires $\hJ\hJ^\dagger=\id$, hence $\hJ\in U(2)$. On the one-photon subspace spanned by $\ket{H}=\hat{a}_H^\dagger\ket{\rm vac}$ and $\ket{V}=\hat{a}_V^\dagger\ket{\rm vac}$, this same matrix acts on the two amplitudes. Replacing $\hJ$ by $e^{i\phi}\hJ$ multiplies a pure state by a global phase and leaves $\hrho\mapsto \hJ\hrho\hJ^\dagger$ unchanged. Since every $U(2)$ matrix can be written as $e^{i\phi}\hU$ with $\hU\in SU(2)$, the physical conjugation action is represented by $SU(2)$.
\end{proof}

For a fixed transmission event, propagation through a fiber of length $L$ is modeled by the spatial Schr\"odinger equation
\begin{eqnarray}
\frac{d}{dz}\ket{\psi(z)}=-i\hH(z)\ket{\psi(z)},
\quad
\hH(z)=\frac{1}{2}\be(z)\cdot\sig,
\label{eq:spatial_schrodinger}
\end{eqnarray}
where $\be(z)\in\R^3$ is the local birefringence vector and $\sig=(\hX,\hY,\hZ)$. The corresponding propagator is
\begin{eqnarray}
\hU(z_2,z_1)=\mathcal{P}\exp\left[-i\int_{z_1}^{z_2}\hH(z)dz\right],
\label{eq:path_ordered}
\end{eqnarray}
with $\hrho(z_2)=\hU(z_2,z_1)\hrho(z_1)\hU(z_2,z_1)^\dagger$. In Bloch-vector form, $\hrho=(\id+\bm r\cdot\sig)/2$, Eq.~\eqref{eq:spatial_schrodinger} is equivalent to
\begin{eqnarray}
\frac{d\bm r(z)}{dz}=\be(z)\times \bm r(z),
\label{eq:bloch_transport}
\end{eqnarray}
so the fiber transports the Stokes vector by a spatially accumulated rotation on the Poincare sphere.

The distinction between a fixed realization and an operational channel will be important. Let $\omega$ denote the uncontrolled classical state of the fiber during a given transmission, including the spatial birefringence profile. For fixed $\omega$, the output is unitary, $\hrho\mapsto \hU_\omega(L,0)\hrho \hU_\omega(L,0)^\dagger$. If the receiver does not condition on $\omega$, the observed map is the random-unitary channel
\begin{eqnarray}
\calE(\hrho)=\E_\omega\left[\hU_\omega(L,0)\hrho \hU_\omega(L,0)^\dagger\right].
\label{eq:uncontrolled_channel}
\end{eqnarray}
This map is completely positive, trace preserving, and unital. The controlled link considered below has the same form with $\hU_\omega$ replaced by the controlled propagator $\hU_{{\rm ctrl},\omega}$.

Divide the fiber into $N$ equal segments of length $\ell=L/N$, with boundaries $z_j=j\ell$. Let
\begin{eqnarray}
\hU_j=\hU(z_j,z_{j-1})
\end{eqnarray}
be the uncontrolled propagator on segment $j$. At the boundaries we insert ideal, zero-length polarization rotators $\hP_j\in\pauli$. The physical controlled propagator is
\begin{eqnarray}
\hU_{\rm ctrl}=\hP_N\hU_N\hP_{N-1}\hU_{N-1}\cdots \hP_1\hU_1\hP_0.
\label{eq:physical_controlled}
\end{eqnarray}
It is more transparent to describe the same sequence by cumulative Pauli frames
\begin{eqnarray}
\hg_j=\hP_j\hP_{j-1}\cdots \hP_0,
\label{eq:frames}
\end{eqnarray}
Here the products retain their actual matrix phases. Their conjugation action depends only on their Pauli equivalence classes. A sequence is called cyclic when $\hg_N=e^{i\chi}\id$ for a known global phase $\chi$, so that the inserted rotators impose no nontrivial Pauli frame at the output. To distinguish a physical propagator from its phase-fixed residual, define
\begin{eqnarray}
\hVop_{\rm ctrl}=\hg_N^\dagger\hU_{\rm ctrl}.
\label{eq:residual_propagator}
\end{eqnarray}
For a cyclic sequence these two propagators induce the same polarization channel. All identity-distance bounds below refer to $\hVop_{\rm ctrl}$, so they are not affected by the overall phase of the physical pulse product.

\begin{proposition}[Toggling-frame factorization]
\label{prop:toggling}
For any embedded Pauli sequence and any segment propagators $\hU_1,\ldots,\hU_N$,
\begin{eqnarray}
\hU_{\rm ctrl}=\hg_N\,\hUtf_N\hUtf_{N-1}\cdots\hUtf_1,
\quad
\hUtf_j=\hg_{j-1}^\dagger \hU_j \hg_{j-1}.
\label{eq:toggling_factor}
\end{eqnarray}
The residual $\hVop_{\rm ctrl}$ is exactly the product of toggling-frame segment propagators. For a cyclic sequence it represents the physical controlled channel after removal of the known global phase.
\end{proposition}

\begin{proof}---Insert $\hg_{j-1}\hg_{j-1}^\dagger=\id$ next to each $\hU_j$ in Eq.~\eqref{eq:physical_controlled} and use $\hg_j=\hP_j\hg_{j-1}$ repeatedly. This gives
\begin{eqnarray}
\hP_N\hU_N\hP_{N-1}\cdots \hP_1\hU_1\hP_0 = \hg_N \hg_{N-1}^\dagger \hU_N \hg_{N-1}\cdots \hg_0^\dagger \hU_1\hg_0,
\end{eqnarray}
which is Eq.~\eqref{eq:toggling_factor}.
\end{proof}

On segment $j$, the exact toggling-frame generator and propagator are
\begin{eqnarray}
\hHtf(z)&=&\hg_{j-1}^\dagger\hH(z)\hg_{j-1},\quad z_{j-1}<z<z_j,\nonumber\\
\hUtf_j&=&\mathcal{P}\exp\left[-i\int_{z_{j-1}}^{z_j}\hHtf(z)dz\right].
\label{eq:toggling_generator}
\end{eqnarray}
The segment average $\hH_j=\ell^{-1}\int_{z_{j-1}}^{z_j}\hH(z)dz$ determines the integrated first moment, but generally $\hU_j\ne\exp(-i\ell\hH_j)$ because generators at distinct positions need not commute. No replacement of the path-ordered propagator by this average-generator exponential is made in the bounds below.

Conjugation by a Pauli flips two components of $\be(z)$. For example,
\begin{eqnarray}
\hX(\beta_x\hX+\beta_y\hY+\beta_z\hZ)\hX=\beta_x\hX-\beta_y\hY-\beta_z\hZ,
\end{eqnarray}
and similarly for $\hY$ and $\hZ$. A Pauli frame sequence therefore modulates the sign pattern of the local birefringence vector. This sign modulation is the mechanism behind the cancellation results below. The four-frame construction and its first-order cancellation have direct precedents in spatial polarization control~\cite{MassarPopescu2007PMD}. The analysis below supplies explicit finite-spacing bounds for the continuous generator, including nonuniform cells, and carries these bounds to a loss- and calibration-error-dependent BB84 certificate.

%-------------------------------------------------------------------------------------------------------------------------------------------------------------------------------------------------------------------------------------
\section{Cancellation theory}
%-------------------------------------------------------------------------------------------------------------------------------------------------------------------------------------------------------------------------------------

The useful regime is one in which the fiber generator is not known but is sufficiently smooth over the spatial scale of an embedded cell. We first record the exact two-segment echo identity, then state the universal Pauli-cell bound that will be used later to control QKD errors.

\begin{theorem}[Exact spatial echo]
\label{thm:echo}
Let $\hU=\exp(-i\hH\ell)$ for a constant Hermitian generator $\hH$. If a Hermitian unitary $\hP$ satisfies $\hP^2=\id$ and $\hP\hH\hP=-\hH$, then
\begin{eqnarray}
\hU\hP\hU\hP=\id.
\label{eq:echo_identity}
\end{eqnarray}
\end{theorem}

\begin{proof}---The anticommutation condition gives
\begin{eqnarray}
\hP\hU\hP=\hP\exp(-i\hH\ell)\hP=\exp(-i\hP\hH\hP\ell)=\exp(i\hH\ell)=\hU^\dagger.
\end{eqnarray}
Therefore, $\hU\hP\hU\hP=\hU(\hP\hU\hP)=\hU\hU^\dagger=\id$.
\end{proof}

The echo is exact but requires a Pauli that flips the particular generator. A universal construction is obtained by using four toggling frames that visit $\pauli$ once per cell.

\begin{lemma}[Pauli symmetrization of traceless generators]
\label{lem:pauli_symm}
For every traceless Hermitian $2\times2$ operator $\hH$,
\begin{eqnarray}
\frac{1}{4}\sum_{\hG\in\pauli} \hG\hH\hG=0.
\label{eq:pauli_average}
\end{eqnarray}
\end{lemma}

\begin{proof}---Write $\hH=(\beta_x\hX+\beta_y\hY+\beta_z\hZ)/2$. Conjugation by $\hX$ leaves the $\hX$ component fixed and flips the $\hY$ and $\hZ$ components; conjugation by $\hY$ and $\hZ$ gives the cyclic variants. Summing the four sign patterns gives
\begin{eqnarray}
(\beta_x,\beta_y,\beta_z) + (\beta_x,-\beta_y,-\beta_z) + (-\beta_x,\beta_y,-\beta_z) + (-\beta_x,-\beta_y,\beta_z)=0,
\end{eqnarray}
so the operator sum vanishes.
\end{proof}

Hence, if $\hH$ is constant over a four-segment cell, the first moment of the toggling-frame generator vanishes. What remains is not generally zero, because the four toggled generators need not commute. The next theorem bounds this noncommuting residual and the additional error caused by spatial variation within the cell.

\begin{lemma}[Unitary remainder bound with controlled first moment]
\label{lem:dyson_main}
Let $\hHtf(z)$ be a measurable Hermitian generator on $[0,T]$ with $\Lambda=\int_0^T\norm{\hHtf(z)}dz<\infty$, and let
\begin{eqnarray}
\hU(T,0)=\mathcal{P}\exp\left[-i\int_0^T\hHtf(z)dz\right],
\quad
\hM_1=\int_0^T\hHtf(z)dz.
\end{eqnarray}
Then,
\begin{eqnarray}
\norm{\hU(T,0)-\id} \le \min\left\{2,\norm{\hM_1}+\frac{\Lambda^2}{2}\right\} \le \norm{\hM_1}+e^\Lambda-1-\Lambda.
\label{eq:dyson_main}
\end{eqnarray}
In particular, if $\hM_1=0$, the deviation is at most $\Lambda^2/2$. No small-$\Lambda$ assumption is required for this inequality.
\end{lemma}

\begin{proof}---The integral evolution equation and unitarity give
\begin{eqnarray}
\hU(T,0)-\id+i\hM_1 &=& -i\int_0^T\hHtf(s)\bigl(\hU(s,0)-\id\bigr)ds, \nonumber\\
\norm{\hU(s,0)-\id} &\le& \int_0^s\norm{\hHtf(r)}dr.
\end{eqnarray}
Consequently,
\begin{eqnarray}
\norm{\hU(T,0)-\id+i\hM_1} \le \int_0^T\norm{\hHtf(s)}\int_0^s\norm{\hHtf(r)}dr\,ds = \frac{\Lambda^2}{2}.
\end{eqnarray}
The triangle inequality, $\norm{\hU-\id}\le2$, and $\Lambda^2/2\le e^\Lambda-1-\Lambda$ establish Eq.~\eqref{eq:dyson_main}.
\end{proof}

\begin{theorem}[End-to-end bound for continuous Pauli-symmetrization cells]
\label{thm:end_to_end}
Consider a cyclic sequence of $M$ four-segment cells. The four segments of cell $m$ have a common length $\ell_m>0$, and $L=4\sum_m\ell_m$. Let $I_{m,j}$ denote its $j$-th segment and assume that the four cumulative frames, modulo global phase, visit $\{\id,\hX,\hY,\hZ\}$ once in each cell. Let $\hH(z)$ be a measurable traceless Hermitian generator and define its exact segment averages by
\begin{eqnarray}
\hH_{m,j}=\frac{1}{\ell_m}\int_{I_{m,j}}\hH(z)dz.
\end{eqnarray}
Suppose that for each cell there exists a traceless Hermitian reference $\bar{\hH}_m$ such that
\begin{eqnarray}
\mathop{\rm ess\,sup}_{z\in {\rm cell}\,m}\norm{\hH(z)}\le h_m,
\quad
\norm{\hH_{m,j}-\bar{\hH}_m}\le\delta_m
\label{eq:cell_assumptions}
\end{eqnarray}
for $j=1,\ldots,4$. Then,
\begin{eqnarray}
\norm{\hVop_{\rm ctrl}-\id} \le \min\left\{2,\sum_{m=1}^M\left(4\delta_m\ell_m+8h_m^2\ell_m^2\right)\right\}.
\label{eq:main_bound_nonuniform}
\end{eqnarray}
For common spacing $\ell_m=\ell$, $h_m\le h$, and $\delta_m\le\delta$, this gives
\begin{eqnarray}
\norm{\hVop_{\rm ctrl}-\id}
&\le&\min\left\{2,L\delta+2Lh^2\ell\right\}\nonumber\\
&\le&\frac{L}{4\ell}\left(4\delta\ell+e^{4h\ell}-1-4h\ell\right).
\label{eq:main_bound_uniform}
\end{eqnarray}
\end{theorem}

\begin{proof}---Let $\hG_{m,j}$ be the actual constant cumulative frame on $I_{m,j}$. Its phase drops out of conjugation. The exact first moment in cell $m$ is
\begin{eqnarray}
\hM_{1,m} &=&\sum_{j=1}^4\hG_{m,j}^\dagger\left(\int_{I_{m,j}}\hH(z)dz\right)\hG_{m,j} \nonumber\\
	&=&\ell_m\sum_{j=1}^4\hG_{m,j}^\dagger\bar{\hH}_m\hG_{m,j} + \ell_m\sum_{j=1}^4\hG_{m,j}^\dagger(\hH_{m,j}-\bar{\hH}_m)\hG_{m,j}.
\end{eqnarray}
The first sum vanishes by Eq.~\eqref{eq:pauli_average}, and the second has norm at most $4\delta_m\ell_m$. Moreover,
\begin{eqnarray}
\Lambda_m=\int_{{\rm cell}\,m}\norm{\hHtf(z)}dz\le4h_m\ell_m.
\end{eqnarray}
Lemma~\ref{lem:dyson_main}, applied to the exact continuous evolution within the cell, yields $\norm{\hU_{{\rm cell},m}-\id}\le4\delta_m\ell_m+8h_m^2\ell_m^2$. The residual in Eq.~\eqref{eq:residual_propagator} is the ordered product of these unitary cell propagators. Telescoping therefore bounds its distance from the identity by the sum of the cell distances, proving Eq.~\eqref{eq:main_bound_nonuniform}. Setting $M=L/(4\ell)$ proves the uniform bound; the last comparison follows from $x^2/2\le e^x-1-x$. The segment averages have entered only the first moment, without approximating any segment propagator.
\end{proof}

\begin{corollary}[Smooth-fiber bound]
\label{cor:lipschitz}
Suppose that, in cell $m$, $\norm{\hH(z)-\hH(z')}\le\kappa_m|z-z'|$ and $\norm{\hH(z)}\le h_m$. Then
\begin{eqnarray}
\norm{\hVop_{\rm ctrl}-\id} \le \min\left\{2,\sum_{m=1}^M(4\kappa_m+8h_m^2)\ell_m^2\right\}.
\label{eq:lipschitz_nonuniform}
\end{eqnarray}
For common spacing and bounds $h_m\le h$, $\kappa_m\le\kappa$,
\begin{eqnarray}
\norm{\hVop_{\rm ctrl}-\id}\le\min\{2,C\ell\},
\quad
C=L(\kappa+2h^2).
\label{eq:lipschitz_bound}
\end{eqnarray}
\end{corollary}

\begin{proof}---Let $c_m$ be the center of cell $m$. The integral of the toggled constant operator $\hH(c_m)$ vanishes over the four equal segments. Hence, the first moment satisfies
\begin{eqnarray}
\norm{\hM_{1,m}} \le \int_{{\rm cell}\,m}\norm{\hH(z)-\hH(c_m)}dz \le \kappa_m\int_{-2\ell_m}^{2\ell_m}|s|ds = 4\kappa_m\ell_m^2.
\end{eqnarray}
Adding $\Lambda_m^2/2\le8h_m^2\ell_m^2$ and telescoping proves the result.
\end{proof}

For fixed $L,h,\kappa$, Eq.~\eqref{eq:lipschitz_bound} gives an explicit $O(L\ell)$ worst-case bound for the continuous generator. It refines the finite-spacing constants without assuming commuting segment generators. The stochastic calculation below concerns typical channel statistics; its interpretation does not require identifying an OU sample path with a Lipschitz function.

%-------------------------------------------------------------------------------------------------------------------------------------------------------------------------------------------------------------------------------------
\section{Effective channel and BB84 consequences}
%-------------------------------------------------------------------------------------------------------------------------------------------------------------------------------------------------------------------------------------

A single realization of the fiber and control sequence is unitary. A user who does not condition on the instantaneous realization sees a random-unitary channel
\begin{eqnarray}
\calE_{\rm ctrl}(\hrho)=\E_\omega\left[\hU_{{\rm ctrl},\omega}\hrho \hU_{{\rm ctrl},\omega}^\dagger\right].
\label{eq:random_unitary}
\end{eqnarray}
This channel is unital. Since the sequence is cyclic, $\hU_{{\rm ctrl},\omega}$ may equivalently be replaced by $\hVop_{{\rm ctrl},\omega}$ in this channel and in its entanglement fidelity. Fig.~\ref{fig:bb84_pauli_toggling_schematic} summarizes how the passive toggling layer enters the BB84 channel model used throughout this section.

\begin{figure}[t]
\centering
\includegraphics[width=0.90\linewidth]{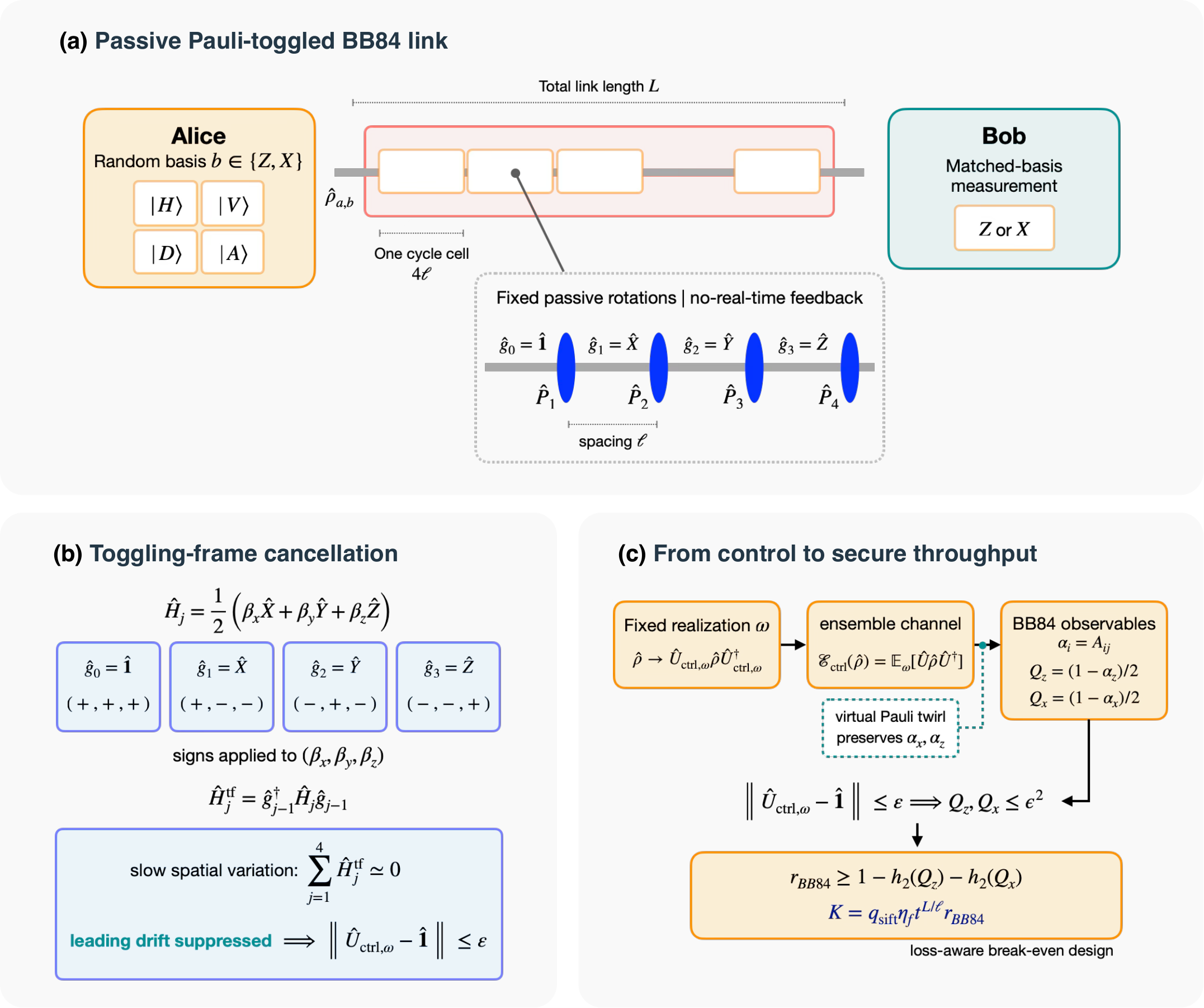}
\caption{Passive spatial Pauli toggling for a polarization-encoded BB84 link. (a) Alice sends a $\hZ$- or $\hX$-basis qubit through a fiber of length $L$ containing cyclic Pauli cells of length $4\ell$, where $\ell$ is the device spacing. The physical rotations $\hP_j$ generate the cumulative frames ${\id,\hX,\hY,\hZ}$ without real-time feedback, after which Bob performs matched-basis detection. (b) In the toggling frame, $\hH_j^{\rm tf}=\hg_{j-1}^{\dagger}\hH_j\hg_{j-1}$, the four frames reverse different components of the local birefringence generator. Sufficiently slow spatial drift therefore cancels to leading order, yielding $\norm{\hVop_{{\rm ctrl},\omega}-\id}\leq\eps$ after removal of the known final phase; the schematic denotes this representative by $\hU_{{\rm ctrl},\omega}$. (c) Averaging the fixed-realization unitary channels produces $\calE_{\rm ctrl}$. A virtual Pauli twirl (dashed) preserves the relevant diagonal Pauli-transfer coefficients and gives $Q_Z=(1-\alpha_z)/2$ and $Q_X=(1-\alpha_x)/2$. These errors determine the asymptotic secret fraction and the loss-aware rate $K=q_{\rm sift}\eta_f t^{L/\ell}r_{\rm BB84}$, exposing the tradeoff between drift suppression and accumulated insertion loss.}
\label{fig:bb84_pauli_toggling_schematic}
\end{figure}

Its action on Bloch vectors is described by the Pauli-transfer matrix
\begin{eqnarray}
A_{ij}=\frac{1}{2}\tr{\hsigma_i\calE_{\rm ctrl}(\hsigma_j)}, \quad i,j \in \{x,y,z\}.
\label{eq:pauli_transfer}
\end{eqnarray}
Equivalently, $A$ is the ensemble average of the $SO(3)$ rotations induced by the unitaries $\hU_{{\rm ctrl},\omega}$.

\begin{proposition}[Average-rotation form]
\label{prop:average_rotation}
Let $R_{\hU_\omega}$ be the Bloch-sphere rotation induced by $\hU_\omega$. For the random-unitary channel $\calE(\hrho)=\E_\omega[\hU_\omega\hrho \hU_\omega^\dagger]$, the Bloch matrix satisfies
\begin{eqnarray}
A=\E_\omega[R_{\hU_\omega}].
\end{eqnarray}
\end{proposition}

\begin{proof}---For a Pauli operator $\hsigma_j$, write $\hU_\omega\hsigma_j\hU_\omega^\dagger=\sum_i(R_{\hU_\omega})_{ij}\hsigma_i$. Linearity of $\calE$ gives $\calE(\hsigma_j)=\sum_i\E_\omega[(R_{\hU_\omega})_{ij}]\hsigma_i$. Taking $\frac12\tr{\hsigma_i(\cdot)}$ yields $A_{ij}=\E_\omega[(R_{\hU_\omega})_{ij}]$.
\end{proof}

For QKD it is convenient to replace $\calE_{\rm ctrl}$ by its Pauli twirl. The twirl preserves the diagonal entries
\begin{eqnarray}
\alpha_i=A_{ii}=\frac{1}{2}\tr{\hsigma_i\calE_{\rm ctrl}(\hsigma_i)},
\end{eqnarray}
and produces a Pauli channel
\begin{eqnarray}
\calP_{\rm ctrl}(\hrho)=p_I\hrho+p_X \hX\hrho \hX+p_Y \hY\hrho \hY+p_Z \hZ\hrho \hZ,
\label{eq:pauli_channel}
\end{eqnarray}
where
\begin{eqnarray}
p_I &=& \frac{1}{4}(1+\alpha_x+\alpha_y+\alpha_z), \nonumber\\
p_X &=& \frac{1}{4}(1+\alpha_x-\alpha_y-\alpha_z), \nonumber\\
p_Y &=& \frac{1}{4}(1-\alpha_x+\alpha_y-\alpha_z), \nonumber\\
p_Z &=& \frac{1}{4}(1-\alpha_x-\alpha_y+\alpha_z).
\label{eq:pauli_probs}
\end{eqnarray}
\begin{proposition}[BB84 QBER from diagonal Pauli-transfer coefficients]
\label{prop:qber_alpha}
For a trace-preserving and unital one-qubit channel, the BB84 error rates in the $\hZ$ and $\hX$ bases are
\begin{eqnarray}
Q_Z(\calE)=\frac{1-\alpha_z}{2},
\quad
Q_X(\calE)=\frac{1-\alpha_x}{2}.
\label{eq:qber_alpha}
\end{eqnarray}
Consequently Pauli twirling preserves $Q_Z$ and $Q_X$.
\end{proposition}

\begin{proof}---We prove the $\hZ$-basis expression. The two $\hZ$-basis states are $\hrho_0=(\id+\hZ)/2$ and $\hrho_1=(\id-\hZ)/2$. If Alice sends $\hrho_0$, Bob's wrong-outcome probability is
\begin{eqnarray}
\tr{\hrho_1\calE(\hrho_0)} = \frac14\tr{(\id-\hZ)(\id+\calE(\hZ))} = \frac12-\frac14\tr{\hZ\calE(\hZ)} = \frac{1-\alpha_z}{2},
\end{eqnarray}
where we used unitality, trace preservation, and $\tr{\hZ}=0$. The same value is obtained for input $\hrho_1$, and averaging over the two bits gives $Q_Z$. The $\hX$-basis result is identical. The twirl preserves $\alpha_x$ and $\alpha_z$ by construction, hence it preserves the two QBERs.
\end{proof}

For the Pauli channel in Eq.~\eqref{eq:pauli_channel}, Eq.~\eqref{eq:qber_alpha} reduces to
\begin{eqnarray}
Q_Z=p_X+p_Y,
\quad
Q_X=p_Z+p_Y.
\label{eq:qber_pauli}
\end{eqnarray}
The $\hY$ error contributes to both bit and phase errors, as expected.

\begin{theorem}[Pauli-channel secret fraction]
\label{thm:pauli_secret_fraction}
For an i.i.d. Bell-diagonal state obtained by applying the Pauli channel in Eq.~\eqref{eq:pauli_channel} to one half of $\ket{\Phi^+}$, knowledge of all four Pauli weights gives the one-way hashing yield
\begin{eqnarray}
r_{\rm Pauli}=\left[1-H_4(p_I,p_X,p_Y,p_Z)\right]_+,
\label{eq:pauli_hashing_rate}
\end{eqnarray}
where $H_4$ is the Shannon entropy of the four weights and $[x]_+=\max\{x,0\}$. For ideal single-photon BB84, using only the bit and phase error rates $e_b=Q_Z$ and $e_p=Q_X$ gives the achievable asymptotic secret-fraction bound
\begin{eqnarray}
r_{\rm BB84} \ge \left[1-h_2(e_b)-h_2(e_p)\right]_+.
\label{eq:bb84_bound}
\end{eqnarray}
Here $h_2(x)=-x\log_2x-(1-x)\log_2(1-x)$, with $h_2(0)=h_2(1)=0$. The BB84 bound assumes asymptotically efficient one-way error correction and the standard ideal source and measurement model.
\end{theorem}

\begin{proof}---The four Pauli operators map Bob's half of $\ket{\Phi^+}=(\ket{00}+\ket{11})/\sqrt2$ to the four Bell states, up to global phase. Their weights determine the hashing yield in Eq.~\eqref{eq:pauli_hashing_rate}~\cite{Bennett1996EntanglementPurification}. Introduce binary variables $B$ and $F$ indicating bit and phase flips. Then $\Pr(B=1)=e_b$, $\Pr(F=1)=e_p$, and $H(B,F)=H_4(p_I,p_X,p_Y,p_Z)$. Thus, $H_4\le h_2(e_b)+h_2(e_p)$. A CSS-based BB84 security argument gives Eq.~\eqref{eq:bb84_bound} from these two marginal errors alone~\cite{Shor2000BB84}; a negative expression is replaced by zero by aborting key extraction.
\end{proof}

The full-weight hashing yield is not determined by the two BB84 error rates alone: the latter do not fix $p_Y$. All key-rate calculations below use the two-error expression in Eq.~\eqref{eq:bb84_bound}, rather than assuming access to the full Pauli distribution. The virtual twirl preserves these error rates but does not assert that the physical channel is Pauli diagonal or replace the usual BB84 security assumptions.

\begin{theorem}[Unitary residual to BB84 performance]
\label{thm:qkd_bound}
Assume that every realization of the cyclic controlled fiber satisfies
\begin{eqnarray}
\norm{\hVop_{{\rm ctrl},\omega}-\id}\le\eps, \quad \eps\ge0.
\end{eqnarray}
Then, the Pauli-twirled channel obeys
\begin{eqnarray}
p_X+p_Y+p_Z \le \min\{1,\eps^2\},
\quad Q_Z,Q_X \le \min\{1,\eps^2\}.
\label{eq:qber_epsilon_bound}
\end{eqnarray}
Define $q_\eps=\min\{\eps^2,1/2\}$. Under the BB84 assumptions of Theorem~\ref{thm:pauli_secret_fraction},
\begin{eqnarray}
r_{\rm BB84}\ge\left[1-h_2(Q_Z)-h_2(Q_X)\right]_+ \ge \left[1-2h_2(q_\eps)\right]_+.
\label{eq:key_fraction_epsilon}
\end{eqnarray}
\end{theorem}

\begin{proof}---The residual $\hVop_{{\rm ctrl},\omega}$ is a product of determinant-one segment propagators, hence belongs to $SU(2)$. For such a unitary, write $\hVop=\cos\theta\,\id-i\sin\theta\,\bm n\cdot\sig$. Its entanglement fidelity relative to the identity is
\begin{eqnarray}
F_e(\hVop)=\abs{\tr{\hVop}}^2/4=\cos^2\theta,
\label{eq:F_e}
\end{eqnarray}
whereas $\bigl\|\hVop-\id\bigr\|=2\abs{\sin(\theta/2)}$~\cite{Schumacher1996EntanglementFidelity}. Therefore,
\begin{eqnarray}
1-F_e(\hVop)=\sin^2\theta\le\min\left\{1,\bigl\|\hVop-\id\bigr\|^2\right\}.
\end{eqnarray}
The entanglement fidelity is linear in the channel and is preserved by the Pauli twirl, so $1-p_I=\E_\omega[1-F_e(\hVop_{{\rm ctrl},\omega})]\le\min\{1,\eps^2\}$. Eq.~\eqref{eq:qber_pauli} bounds each QBER by $1-p_I$. If $\eps^2\le1/2$, monotonicity of $h_2$ on $[0,1/2]$ proves Eq.~\eqref{eq:key_fraction_epsilon}. If $\eps^2>1/2$, $h_2(Q)\le1=h_2(q_\eps)$ for every $Q\in[0,1]$, and the stated lower bound is zero. This also proves the assertion in the noninformative regime.
\end{proof}

The entropy bound gives a strictly positive certificate only when $\eps^2<q_{\rm th}$, where $q_{\rm th}\in(0,1/2)$ is the unique solution of $h_2(q_{\rm th})=1/2$. Requiring merely $\eps<1$ is insufficient, because binary entropy decreases beyond $1/2$.

\begin{corollary}[Ensemble and high-probability certificates]
\label{cor:ensemble}
If $B_\omega\ge\norm{\hVop_{{\rm ctrl},\omega}-\id}$ almost surely, then
\begin{eqnarray}
1-p_I,\ Q_Z,\ Q_X\le q_{\rm ens},
\quad
q_{\rm ens}=\E_\omega\left[\min\{1,B_\omega^2\}\right].
\label{eq:ensemble_bound}
\end{eqnarray}
Alternatively, if $\Pr(\norm{\hVop_{{\rm ctrl},\omega}-\id}\le\eps)\ge1-\zeta$, with $0\le\zeta\le1$, then
\begin{eqnarray}
1-p_I,\ Q_Z,\ Q_X\le q_{\rm hp},
\quad
q_{\rm hp}=(1-\zeta)\min\{1,\eps^2\}+\zeta.
\label{eq:high_probability_bound}
\end{eqnarray}
For either bound $q$, the corresponding secret-fraction certificate is $[1-2h_2(\min\{q,1/2\})]_+$.
\end{corollary}

\begin{proof}---Average the pointwise entanglement-fidelity inequality in the proof of Theorem~\ref{thm:qkd_bound} to obtain Eq.~\eqref{eq:ensemble_bound}. On the event of probability at least $1-\zeta$, use $1-F_e\le\min\{1,\eps^2\}$; on its complement use $1-F_e\le1$. This proves Eq.~\eqref{eq:high_probability_bound}. The entropy step is the same as in Theorem~\ref{thm:qkd_bound}.
\end{proof}

These statements allow realization-dependent generator bounds; no common finite bound over a Gaussian ensemble is assumed. The probability $\zeta$ refers to the modeled physical ensemble, not to a finite-key security parameter. The expectation or tail bound must be established for that ensemble before using the certificate quantitatively.

The total key rate per launched pulse also includes loss. For unbiased BB84, let $q_{\rm sift}=1/2$. Let $0<\eta_f\le1$ denote the baseline transmission-and-detection probability excluding the embedded devices, and let each device have power transmission $t$. For the canonical sequence, take $\hP_0=\id$ and count the $N(\ell)=L/\ell$ physical boundary operations, including the final closing rotation. With this convention, the multiplicative insertion factor is $\eta_{\rm ins}=t^{N(\ell)}$, and
\begin{eqnarray}
K(\ell,t)=q_{\rm sift}\eta_f t^{L/\ell}\left[1-h_2(Q_Z(\ell))-h_2(Q_X(\ell))\right]_+,
\label{eq:key_rate_loss}
\end{eqnarray}
This expression defines the asymptotic BB84 rate used for the comparisons; it is not an optimized channel capacity. The loss is assumed to be polarization independent and independent of the fiber realization, so conditioning on detection does not reweight the modeled channel ensemble. A controlled design beats the uncontrolled key rate $K_0$ exactly when the QBER-induced increase in secret fraction is larger than the insertion-loss penalty.

\begin{proposition}[Break-even condition]
\label{prop:break_even}
Let $Q_Z^0,Q_X^0$ be the uncontrolled BB84 error rates and suppose that the controlled spacing $\ell$ gives $Q_Z(\ell),Q_X(\ell)$. In the positive-secret-fraction regime, $K(\ell,t)>K_0$ if and only if
\begin{eqnarray}
t^{L/\ell} > \frac{1-h_2(Q_Z^0)-h_2(Q_X^0)}{1-h_2(Q_Z(\ell))-h_2(Q_X(\ell))}.
\label{eq:break_even}
\end{eqnarray}
In the symmetric case $Q_Z^0=Q_X^0=Q_0$ and $Q_Z(\ell)=Q_X(\ell)=Q(\ell)$, this reduces to
\begin{eqnarray}
t^{L/\ell} > \frac{1-2h_2(Q_0)}{1-2h_2(Q(\ell))}.
\end{eqnarray}
\end{proposition}

\begin{proof}---Divide Eq.~\eqref{eq:key_rate_loss} by the uncontrolled expression $K_0=q_{\rm sift}\eta_f[1-h_2(Q_Z^0)-h_2(Q_X^0)]$. The common factor $q_{\rm sift}\eta_f$ cancels, and the condition that the ratio exceeds one is exactly Eq.~\eqref{eq:break_even}.
\end{proof}

Combining Theorems~\ref{thm:end_to_end} and~\ref{thm:qkd_bound} gives a conservative design rule. Let $\eps(\ell)$ be a uniform residual bound, for example that in Eq.~\eqref{eq:main_bound_uniform} or~\eqref{eq:lipschitz_bound}. Then,
\begin{eqnarray}
&& Q_Z(\ell),Q_X(\ell) \le \min\{1,\eps(\ell)^2\}, \nonumber\\
&& K(\ell,t) \ge q_{\rm sift}\eta_f t^{L/\ell} \left[1-2h_2\!\left(\min\{\eps(\ell)^2,1/2\}\right)\right]_+.
\label{eq:conservative_design_rule}
\end{eqnarray}
For a probabilistic generator model, $\eps(\ell)^2$ inside the entropy certificate may instead be replaced by a justified $q_{\rm ens}$ or $q_{\rm hp}$ from Corollary~\ref{cor:ensemble}. These lower bounds can be conservative; a zero certificate does not imply a zero physical key rate. Directly estimated $Q_Z,Q_X$ give the empirical comparison in Proposition~\ref{prop:break_even}.

\begin{proposition}[Certified gain and a loss exclusion condition]
\label{prop:certified_gain}
Let $r_0=[1-h_2(Q_Z^0)-h_2(Q_X^0)]_+>0$, and let $T$ be the total polarization-independent insertion transmission. If a residual or ensemble bound certifies $r_{\rm cert}$ for the controlled link, then
\begin{eqnarray}
T r_{\rm cert}>r_0 \quad\Longrightarrow\quad K_{\rm ctrl}>K_0.
\label{eq:certified_gain}
\end{eqnarray}
Conversely, $T\le r_0$ excludes a strict gain in the rate model~\eqref{eq:key_rate_loss}, regardless of how strongly the polarization errors are suppressed. For identical devices, a necessary condition is
\begin{eqnarray}
t^N > r_0,
\end{eqnarray}
and, when $r_{\rm cert}>r_0$, a sufficient condition is $t>(r_0/r_{\rm cert})^{1/N}$.
\end{proposition}

\begin{proof}---The certified rate obeys $K_{\rm ctrl}\ge q_{\rm sift}\eta_f T r_{\rm cert}$ and $K_0=q_{\rm sift}\eta_f r_0$, proving the sufficient condition. Since the modeled secret fraction is at most one, $K_{\rm ctrl}\le q_{\rm sift}\eta_f T$; this proves the exclusion condition. The identical-device statements follow by setting $T=t^N$.
\end{proof}

\begin{corollary}[A sufficient spacing interval for smooth fibers]
\label{cor:spacing_window}
Assume the common smooth-fiber bounds of Corollary~\ref{cor:lipschitz} for every realization, with $C=L(\kappa+2h^2)>0$, ideal rotations, and $0<t<1$. Choose a target error $0<q_\star<q_{\rm th}$ such that $r_\star=1-2h_2(q_\star)>r_0>0$. Every physically admissible spacing $\ell=L/(4M)$, $M\in\mathbb{N}$, satisfying
\begin{eqnarray}
\frac{L(-\ln t)}{\ln(r_\star/r_0)}<\ell\le\frac{\sqrt{q_\star}}{C}
\label{eq:certified_spacing_window}
\end{eqnarray}
certifies $K(\ell,t)>K_0$.
\end{corollary}

\begin{proof}---The upper inequality gives $\eps(\ell)^2\le C^2\ell^2\le q_\star$, hence $r_{\rm cert}\ge r_\star$. The lower inequality is equivalent to $t^{L/\ell}r_\star>r_0$. Proposition~\ref{prop:certified_gain} completes the proof.
\end{proof}

Eq.~\eqref{eq:certified_spacing_window} makes the competing requirements explicit: control sets an upper spacing, whereas accumulated insertion loss sets a lower spacing. It is a sufficient interval and must contain an admissible integer cell count; an empty interval does not exclude an empirical gain. This is an analytical certificate, not a fit to the numerical optimum.

%-------------------------------------------------------------------------------------------------------------------------------------------------------------------------------------------------------------------------------------
\section{Numerical demonstration with spatially correlated birefringence}
%-------------------------------------------------------------------------------------------------------------------------------------------------------------------------------------------------------------------------------------

We now test the scheme on a stochastic fiber model. The calculation does not use the analytical bound as an input; it directly multiplies the path-ordered $SU(2)$ propagators and then computes the channel statistics. The birefringence vector is modeled as a stationary spatial Ornstein--Uhlenbeck (OU) process,
\begin{eqnarray}
d\be(z)=-\frac{\be(z)-\be_0}{\xi}dz+\sqrt{\frac{2\sigma_\beta^2}{\xi}}\,d\bm W_z,
\label{eq:ou_model}
\end{eqnarray}
with independent components, correlation length $\xi$, rms fluctuation $\sigma_\beta$ for each component, and small mean bias $\be_0$. The illustrative model parameters used in the figures are
\begin{eqnarray}
L=50~{\rm km},
\quad
\xi=3~{\rm km},
\quad
\sigma_\beta=0.020~{\rm rad/km},
\quad
\be_0=(0.002,-0.001,0.0015)~{\rm rad/km}.
\end{eqnarray}
Each realization is discretized into 2400 fine spatial steps. For a chosen device spacing $\ell=L/N$, with $N$ a multiple of four, the controlled propagator is computed with the cyclic frame order $\id,\hX,\hY,\hZ$. We average over 200 fiber realizations. The QKD parameters are $q_{\rm sift}=1/2$, $\eta_f=0.1$, and representative per-device transmissions $t=0.997$ and $t=0.995$.

The OU paths are continuous and bounded on a finite link almost surely, so the realization-dependent continuous-cell bound is well defined. They are not Lipschitz almost surely, and their Gaussian ensemble has no finite common deterministic generator bound. Thus, these simulations do not directly test the smooth-fiber corollary. Corollary~\ref{cor:ensemble} describes how a separately justified ensemble or tail estimate would give a probabilistic certificate; no such estimate is substituted into the numerical rates reported here. The chosen $\xi$ and $\sigma_\beta$ specify an effective spatial-noise model, rather than parameters inferred from a measured deployed fiber.

The cumulative-product step in Algorithm~\ref{alg:MCE} makes the spacing scan efficient: the fine stochastic path is generated once per realization, and each coarse segment is recovered by a unitary quotient. This also ensures that different spacings are compared on the same ensemble of fiber realizations rather than on independently sampled fibers.

\begin{algorithm}[H]
\caption{Monte Carlo evaluation of a spatial Pauli-toggled fiber channel}
\label{alg:MCE}
\begin{algorithmic}[1]
\State Choose $L$, fine step $\Delta z$, OU parameters $(\be_0,\sigma_\beta,\xi)$, and a list of segment counts $N$.
\State For each realization $\omega=1,\ldots,N_{\rm MC}$, generate a spatial OU path $\be_\omega(z_k)$.
\State Form fine-step unitaries $\hU_k=\exp[-i\Delta z\,\be_\omega(z_k)\cdot\sig/2]$ and cumulative products $\hC_m=\hU_{m-1}\cdots \hU_0$.
\State Store the uncontrolled unitary $\hU_{\rm un}=\hC_{n_{\rm fine}}$.
\State For each allowed segment count $N$, compute coarse segment propagators as $\hU(z_j,z_{j-1})=\hC_{b_j}\hC_{a_j}^\dagger$, conjugate each segment by the cyclic frame $\hg_{j-1}\in(\id,\hX,\hY,\hZ)$, and multiply the toggling-frame product.
\State Convert each resulting unitary into $Q_Z$, $Q_X$, and $F_e=\abs{\tr{\hU}}^2/4$; average the channel statistics and compute $K(\ell,t)$.
\end{algorithmic}
\end{algorithm}

\begin{figure}[t]
\centering
\includegraphics[width=0.90\linewidth]{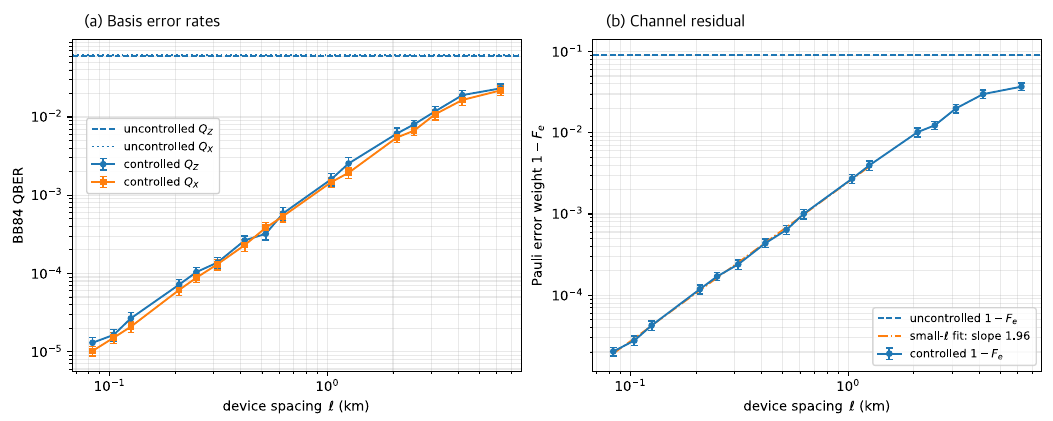}
\caption{Suppression of polarization-induced channel errors in the spatial OU fiber model. (a) Ensemble-averaged BB84 QBER in the $\hZ$ and $\hX$ bases as a function of device spacing $\ell$. Horizontal dashed/dotted lines show the uncontrolled channel. (b) Total Pauli error weight $1-F_e$ obtained from the entanglement fidelity. The small-$\ell$ fit is a visual diagnostic of typical scaling in this stochastic model; it is not imposed by the theorem. Error bars are 95\% confidence intervals over 200 realizations.}
\label{fig:suppression}
\end{figure}

Fig.~\ref{fig:suppression} shows the main physical effect. The uncontrolled fiber has average QBERs $Q_Z\simeq0.060$ and $Q_X\simeq0.062$. With Pauli toggling at $\ell=1.25$ km ($N=40$ segments), the same model gives $Q_Z\simeq2.54\times10^{-3}$ and $Q_X\simeq1.93\times10^{-3}$. The total Pauli error weight decreases from about $9.1\times10^{-2}$ to about $4.0\times10^{-3}$. For smaller spacings the residual continues to fall, but eventually the marginal benefit for key generation is outweighed by insertion loss.

\begin{table}[t]
\begin{adjustbox}{width=0.70\columnwidth}
\setlength{\tabcolsep}{15pt}
\begin{tabular}{lccccc}
\toprule
case & $N$ & $\ell$ (km) & $\bar Q$ & $r$ & $K$ \\
\midrule
uncontrolled & 0 & -- & $6.11\times10^{-2}$ & $0.336$ & $1.68\times10^{-2}$ \\
controlled & 16 & 3.125 & $1.13\times10^{-2}$ & $0.821$ & $3.91\times10^{-2}$ \\
controlled & 40 & 1.250 & $2.24\times10^{-3}$ & $0.954$ & $4.23\times10^{-2}$ \\
controlled & 120 & 0.417 & $2.47\times10^{-4}$ & $0.993$ & $3.46\times10^{-2}$ \\
controlled & 240 & 0.208 & $6.68\times10^{-5}$ & $0.998$ & $2.43\times10^{-2}$ \\
controlled & 480 & 0.104 & $1.57\times10^{-5}$ & $0.999$ & $1.18\times10^{-2}$ \\
\bottomrule
\end{tabular}
\end{adjustbox}
\caption{Representative simulation metrics. Here, $\bar Q=(Q_Z+Q_X)/2$, $r=1-h_2(Q_Z)-h_2(Q_X)$, and $K=\frac12\eta_f t^N r$ with $\eta_f=0.1$ and $t=0.997$.}
\label{tab:simulation_metrics}
\end{table}

Tab.~\ref{tab:simulation_metrics} and Fig.~\ref{fig:keyrate} translate the channel improvement into key rates. At $t=0.997$, the optimum in this scan occurs near $N=40$ devices, where the key rate is approximately $4.23\times10^{-2}$ bits per launched pulse, about 2.5 times the uncontrolled rate. For a more lossy device with $t=0.995$, the optimum shifts to a larger spacing, near $N=24$, because the loss penalty grows faster as the sequence is densified.

\begin{figure}[t]
\centering
\includegraphics[width=0.90\linewidth]{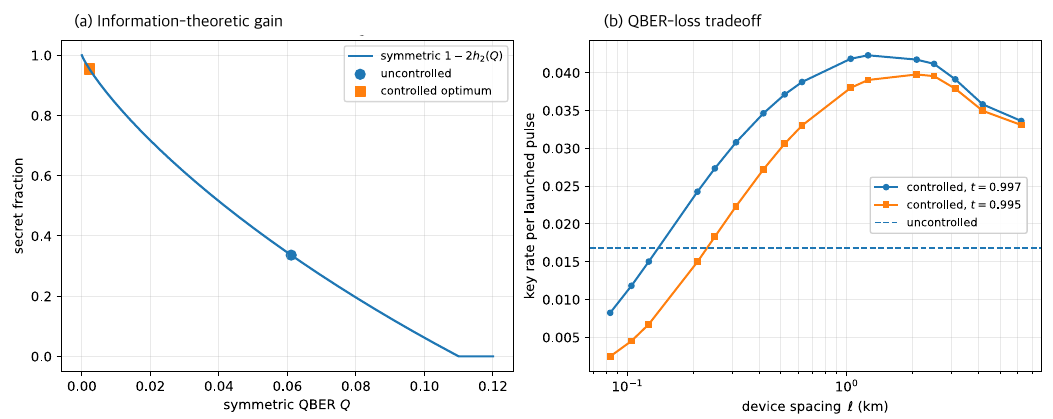}
\caption{BB84 information gain and insertion-loss tradeoff. (a) The symmetric-error secret fraction $1-2h_2(Q)$, with markers indicating the uncontrolled channel and the controlled optimum for $t=0.997$. (b) Key rate per launched pulse versus spacing for two per-device transmissions. The horizontal line is the uncontrolled key rate.}
\label{fig:keyrate}
\end{figure}

\begin{figure}[t]
\centering
\includegraphics[width=0.60\linewidth]{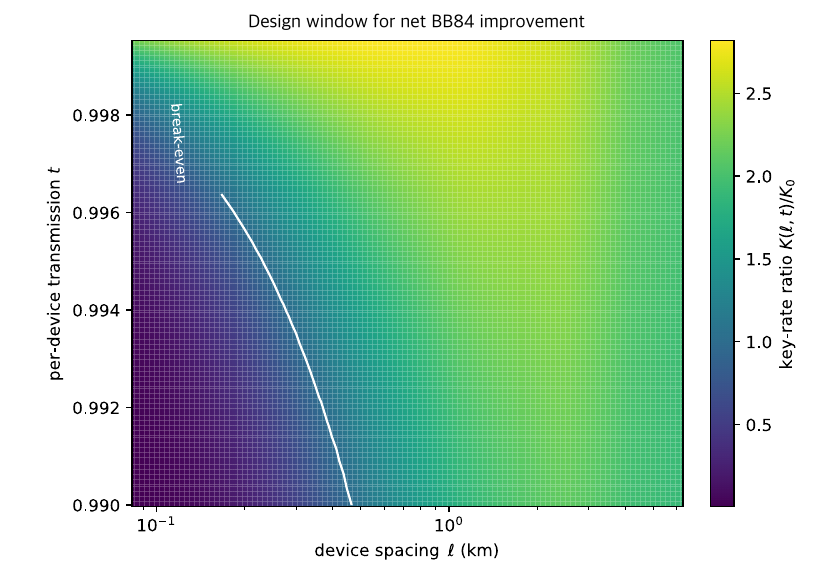}
\caption{Design map for net BB84 improvement. The color scale is $K(\ell,t)/K_0$, where $K_0$ is the uncontrolled key rate. The white contour is the break-even line. Above this contour, the QBER reduction produced by spatial Pauli toggling compensates for the insertion loss.}
\label{fig:design_map}
\end{figure}

Fig.~\ref{fig:design_map} summarizes the engineering message. There is a finite design window, not a monotone rule that ``more devices are always better.'' For high-transmission devices, decreasing $\ell$ initially improves the rate because it suppresses coherent polarization errors. Once the QBER is already small, further densification only multiplies loss and reduces the launched-pulse key rate. This behavior is consistent with Eq.~\eqref{eq:key_rate_loss} and provides a practical way to choose $\ell$ once device loss and measured polarization drift statistics are known.

As a second numerical check, Fig.~\ref{fig:xi_sensitivity} scans the spatial correlation length while keeping the total fiber length and birefringence rms fixed. The uncontrolled QBER grows with $\xi$ in this parameter range because larger correlated domains accumulate more coherent rotation before changing direction. The controlled curves behave differently: for the denser sequence ($N=120$), the mean QBER remains below $4\times10^{-3}$ across the scan; for $N=40$, it remains at the $10^{-3}$ scale for long correlation lengths and increases only when the correlation length becomes comparable to or shorter than the cell scale. This is consistent with the physical cancellation mechanism: the Pauli cell is most effective when the generator is approximately constant over a cell, and it becomes less efficient when spatial variation is rapid.

\begin{figure}[t]
\centering
\includegraphics[width=0.60\linewidth]{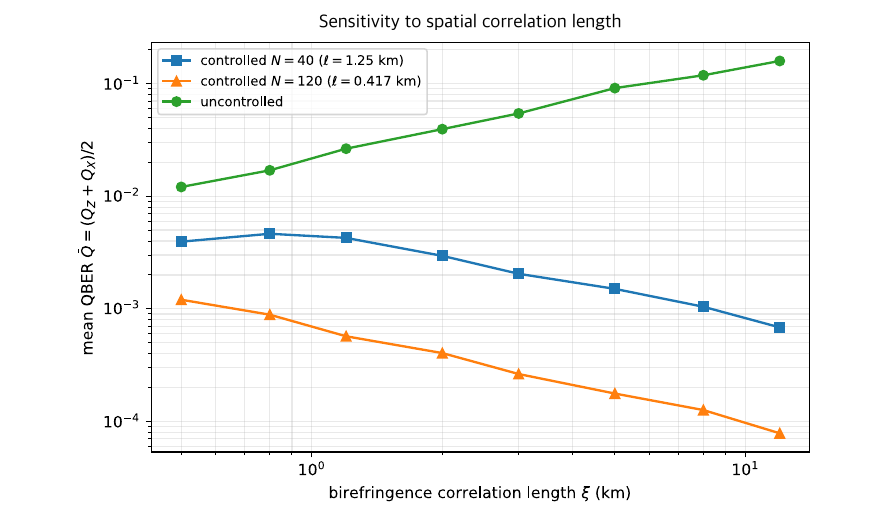}
\caption{Sensitivity to the birefringence correlation length $\xi$. The total length, rms birefringence, and mean bias are fixed; only $\xi$ is varied. The controlled curves use the same cyclic Pauli-cell sequence with $N=40$ and $N=120$ segments. The scan confirms that spatial Pauli toggling is strongest when the generator is correlated over the cell scale.}
\label{fig:xi_sensitivity}
\end{figure}

The correlation-length scan illustrates why the deterministic bound and the Monte Carlo calculation answer different questions. The bound is uniform over generators satisfying its stated local assumptions, whereas the simulation samples the specified OU ensemble. Moreover, the small-spacing fit in Fig.~\ref{fig:suppression} concerns $1-F_e$, not the operator-norm residual. A linear norm bound already implies a quadratic upper bound on the Pauli error through Theorem~\ref{thm:qkd_bound}; a fitted exponent near two therefore does not by itself establish a stronger order of suppression. The rate improvement is relative to an uncontrolled link. Neither performance against an active tracker nor its required update bandwidth is evaluated by this spatial ensemble.

%-------------------------------------------------------------------------------------------------------------------------------------------------------------------------------------------------------------------------------------
\section{Practical interpretation and extensions}
%-------------------------------------------------------------------------------------------------------------------------------------------------------------------------------------------------------------------------------------

The theoretical scheme assumes ideal Pauli rotations. In a polarization implementation, these can be realized by fixed wave-plate combinations or by equivalent fiber-based polarization controllers set to implement $\pi$ rotations on the Poincare sphere. The important requirement is not that the optical element be a matrix called $\hX$, $\hY$, or $\hZ$ in a laboratory basis, but that its conjugation action flip the corresponding Stokes components. A global phase and an overall sign are immaterial to the polarization channel; the residual convention in Eq.~\eqref{eq:residual_propagator} removes them before taking a norm. This makes the scheme compatible with a variety of passive optical layouts, provided the element axes are stable relative to the chosen polarization basis.

The frame description also gives a direct recipe for the physical pulse sequence. For example, the cell used in the simulations has segment frames $(\hg_0,\hg_1,\hg_2,\hg_3)\sim(\id,\hX,\hY,\hZ)$ and returns to $\hg_4\sim\id$, where $\sim$ denotes equality up to global phase. Since the actual boundary operation is $\hP_j=\hg_j\hg_{j-1}^\dagger$, the required pulse pattern is shown in Tab.~\ref{tab:cell_pulses}. Products are listed modulo global phase, so $\hY\hX$ is implemented as $\hZ$ and $\hZ\hY$ as $\hX$. Other permutations of the four frames are possible; they have the same first-moment cancellation but can have different higher-order commutator terms. This leaves room for engineering variants, such as choosing the frame order to reduce known dominant commutators or to match available optical hardware.

\begin{table}[t]
\begin{adjustbox}{width=0.70\columnwidth}
\setlength{\tabcolsep}{15pt}
\begin{tabular}{cccc}
\toprule
boundary/segment & desired frame & boundary pulse & sign pattern on $(\beta_x,\beta_y,\beta_z)$ \\
\midrule
segment 1 & $\hg_0=\id$ & $\hP_0=\id$ & $(+,+,+)$ \\
segment 2 & $\hg_1=\hX$ & $\hP_1=\hX$ & $(+,-,-)$ \\
segment 3 & $\hg_2=\hY$ & $\hP_2=\hY\hX\sim \hZ$ & $(-,+,-)$ \\
segment 4 & $\hg_3=\hZ$ & $\hP_3=\hZ\hY\sim \hX$ & $(-,-,+)$ \\
cell close & $\hg_4=\id$ & $\hP_4=\hZ$ & -- \\
\bottomrule
\end{tabular}
\end{adjustbox}
\caption{Canonical four-frame Pauli cell. The segment frame $\hg_{j-1}$ determines the sign pattern seen by the $j$-th fiber segment. The frame labels and boundary pulses are displayed modulo global phase; exact matrix products use the convention of Eq.~\eqref{eq:residual_propagator}.}
\label{tab:cell_pulses}
\end{table}

Device imperfections enter through both insertion loss and coherent calibration error. The former is included in Eq.~\eqref{eq:key_rate_loss}; the latter admits a deterministic bound without a first-order expansion.

\begin{proposition}[Accumulated rotation-error bound]
\label{prop:device_error}
Let the $N$ nonidentity boundary operations be implemented as $\widehat P_j^{\rm imp}=\hP_j e^{-i\hDelta_j}$, where, after discarding global phases, $\hDelta_j$ is traceless Hermitian and $\norm{\hDelta_j}\le\eta_j$. Let $\hU_{\rm ctrl}^{\rm imp}$ and $\hU_{\rm ctrl}^{\rm id}$ be the physical propagators on the same fiber realization. Define both residuals with the ideal cumulative frame, $\hVop_{\rm ctrl}^{\rm imp}=\hg_N^\dagger\hU_{\rm ctrl}^{\rm imp}$ and $\hVop_{\rm ctrl}^{\rm id}=\hg_N^\dagger\hU_{\rm ctrl}^{\rm id}$. If $\norm{\hVop_{\rm ctrl}^{\rm id}-\id}\le\eps_0$, then
\begin{eqnarray}
\norm{\hVop_{\rm ctrl}^{\rm imp}-\id} \le \min\left\{2,\eps_0+\sum_{j=1}^N\gamma_j\right\},
\quad
\gamma_j = 2\sin\left(\frac{\min\{\eta_j,\pi\}}{2}\right)\le\min\{\eta_j,2\}.
\label{eq:device_error_bound}
\end{eqnarray}
No independence or zero-mean assumption is imposed on the calibration errors.
\end{proposition}

\begin{proof}---The spectral theorem gives $\norm{e^{-i\hDelta_j}-\id}\le\gamma_j$. Replace the imperfect operations by the ideal ones one at a time in the full physical product. All surrounding fiber propagators and rotations are unitary, so each replacement changes the product by at most $\gamma_j$ in norm. Thus, $\norm{\hU_{\rm ctrl}^{\rm imp}-\hU_{\rm ctrl}^{\rm id}}\le\sum_j\gamma_j$. Multiplication by the same $\hg_N^\dagger$, followed by the triangle inequality, proves Eq.~\eqref{eq:device_error_bound}. The residuals remain in $SU(2)$ under the stated traceless-error convention.
\end{proof}

When the stated residual and device-error bounds hold for every realization, and the scalar device transmissions $t_j$ are realization independent, set $T=\prod_jt_j$ and $\eps_{\rm imp}=\eps_0+\sum_j\gamma_j$. The certified imperfect-device rate satisfies
\begin{eqnarray}
K_{\rm imp}\ge q_{\rm sift}\eta_f T \left[1-2h_2\!\left(\min\{\eps_{\rm imp}^2,1/2\}\right)\right]_+.
\label{eq:imperfect_rate_certificate}
\end{eqnarray}
This bound covers systematic as well as stochastic bounded rotation errors; it need not be saturated by either. If the loss depends on the fiber realization, the channel average must instead be taken with the detection-conditioned ensemble before applying a QBER certificate.

\begin{corollary}[Spacing interval with bounded device errors]
\label{cor:imperfect_window}
Assume the smooth-fiber coefficient $C>0$ of Eq.~\eqref{eq:lipschitz_bound}, identical transmissions $0<t<1$, and $\gamma_j\le\gamma$. For the target $q_\star,r_\star,r_0$ of Corollary~\ref{cor:spacing_window}, suppose $q_\star\ge4C\gamma L$ and define
\begin{eqnarray}
\ell_\pm=\frac{\sqrt{q_\star}\pm\sqrt{q_\star-4C\gamma L}}{2C}.
\label{eq:rotation_window_endpoints}
\end{eqnarray}
Every admissible spacing $\ell=L/(4M)$ satisfying
\begin{eqnarray}
\ell_-\le\ell\le\ell_+,
\quad
\ell>\frac{L(-\ln t)}{\ln(r_\star/r_0)}
\label{eq:imperfect_spacing_window}
\end{eqnarray}
certifies $K_{\rm imp}>K_0$.
\end{corollary}

\begin{proof}---Taking $\eps_0=\min\{2,C\ell\}$, Corollary~\ref{cor:lipschitz} and Proposition~\ref{prop:device_error} imply $\eps_{\rm imp}\le C\ell+\gamma L/\ell$. The inequality $C\ell+\gamma L/\ell\le\sqrt{q_\star}$ is equivalent to $C\ell^2-\sqrt{q_\star}\ell+\gamma L\le0$, whose roots are Eq.~\eqref{eq:rotation_window_endpoints}. Thus, the first condition certifies $r_{\rm cert}\ge r_\star$, while the second gives $t^{L/\ell}r_\star>r_0$.
\end{proof}

For $\gamma>0$, the upper-bound expression $C\ell+\gamma L/\ell$ has minimum $2\sqrt{C\gamma L}$ at $\ell=\sqrt{\gamma L/C}$ before imposing the integer cell count and insertion loss. This is the minimizer of the residual certificate, not a prediction of the measured key-rate optimum. Eqs.~\eqref{eq:rotation_window_endpoints}--\eqref{eq:imperfect_spacing_window} show analytically how bounded device error can close a sufficient operating interval even when optical transmission is high.

The existing ideal-rotation values in Tab.~\ref{tab:simulation_metrics} also give a direct insertion-loss tolerance without an additional simulation. Using $N=40$, $r_0=0.336$, and $r_c=0.954$, Proposition~\ref{prop:break_even} gives
\begin{eqnarray}
t_{\rm min}=\left(\frac{0.336}{0.954}\right)^{1/40}\simeq0.97425,
\quad
a_{\rm max}=\frac{10}{40}\log_{10}\left(\frac{0.954}{0.336}\right)\simeq0.1133~{\rm dB},
\label{eq:loss_tolerance_existing}
\end{eqnarray}
where $a=-10\log_{10}t$ is the loss per installed device and strict improvement requires $a<a_{\rm max}$. The illustrative $t=0.997$ corresponds to about $0.0130$ dB per device. The tolerance in Eq.~\eqref{eq:loss_tolerance_existing} is conditional on retaining the tabulated ideal-rotation QBERs; a device-error guarantee instead uses Eq.~\eqref{eq:imperfect_rate_certificate}. An installed-device loss budget must include coupling, splices or connectors, and packaging as well as the rotator itself.

The present frequency-independent model excludes polarization-dependent loss and polarization-mode dispersion (PMD). A polarization-only unitary cannot restore an already lost photon or purify a mixed polarization state after frequency has been traced out. This does not forbid interleaved rotations from suppressing the formation of coherent polarization--frequency correlations during propagation. For example, $\hU(\nu)=\exp[-i\nu\tau\hZ/2]$ obeys $\hU(\nu)\hX\hU(\nu)\hX=\id$ at every angular-frequency detuning $\nu$, where $\tau$ is the differential group delay of each identical segment and the rotation is ideal and achromatic. Such PMD refocusing requires a frequency-dependent propagation and device model~\cite{MassarPopescu2007PMD}, which is outside the present bounds. Statements about uncorrected PMD and spectral correlations refer to this restricted model, rather than a general impossibility of coherent refocusing. Polarization-dependent attenuation additionally requires a loss-aware security model and cannot be represented by the scalar transmission factor alone.

An experimentally natural characterization procedure is as follows. First measure the uncontrolled Pauli-transfer diagonal entries by sending calibration states in the $\hX$, $\hY$, and $\hZ$ bases, or equivalently estimate the BB84 test-basis error rates over time. For a unital polarization channel, the calibration relation is particularly simple:
\begin{eqnarray}
\widehat\alpha_z=1-2\widehat Q_Z,
\quad
\widehat\alpha_x=1-2\widehat Q_X,
\end{eqnarray}
and an optional $\hY$-basis calibration gives $\widehat\alpha_y$. These three diagonal coefficients determine the Pauli-twirled error probabilities through Eq.~\eqref{eq:pauli_probs}. Then insert one Pauli-cell period and scan the cell spacing over a small number of values. The observed $Q_Z(\ell)$ and $Q_X(\ell)$ can be inserted directly into Eq.~\eqref{eq:key_rate_loss}, while the measured transmission gives $t^{L/\ell}$ without assuming a per-device loss model. This procedure does not require reconstructing the full stochastic process $\be(z)$. The OU model in this paper is therefore best regarded as a reproducible stress test and a guide for parameter ranges, not as a prerequisite for deployment.

The same data can be used to distinguish three regimes. In the first regime, the sequence lowers QBER and the measured insertion factor remains above the break-even threshold in Eq.~\eqref{eq:break_even}; the passive sequence improves the launched-pulse key rate. In the second regime, the QBER decreases but the insertion factor is too small; the sequence may still be useful in a system where the limiting resource is error correction overhead or stability rather than launched-pulse rate, but it does not increase the simple rate metric used here. In the third regime, the QBER fails to decrease substantially. This indicates either that the fiber generator varies too rapidly relative to the cell length, that the dominant impairment is nonunitary, or that the inserted rotations are not calibrated in the intended Pauli frame.

The framework also suggests two variants. The first is a sparse hybrid design in which passive Pauli cells could reduce the drift presented to an active polarization tracker; a larger update interval would require a separate temporal model or time-series validation. The second is an adaptive manufacturing design: the positions of passive rotators are fixed after installation but chosen from a one-time characterization of the link, allowing nonuniform spacings in sections with different stress or bend profiles. The proof of Theorem~\ref{thm:end_to_end} already permits cell-dependent parameters $h_m$ and $\delta_m$, so a nonuniform spacing analysis follows by replacing the uniform bound with the nonuniform sum in Eq.~\eqref{eq:main_bound_nonuniform}.

Finally, the same toggling principle is not limited to BB84. Any protocol whose performance is degraded by an unknown polarization rotation can benefit from a channel closer to the identity. For entanglement-based QKD, the same Pauli probabilities determine Bell-state weights. For measurement-device-independent or decoy-state variants, the single-photon polarization channel appears inside a larger optical model, but the coherent polarization part is still described by the same $SU(2)$ process. The main extra work in those settings is to combine the polarization improvement with source statistics, multi-photon terms, and finite-size estimation.

%-------------------------------------------------------------------------------------------------------------------------------------------------------------------------------------------------------------------------------------
\section{Design workflow for an experimental link}
%-------------------------------------------------------------------------------------------------------------------------------------------------------------------------------------------------------------------------------------

The previous sections give a mathematical bound, a channel-level conversion, and a Monte Carlo example. For an actual link, endpoint polarization measurements first establish the observed QBERs, attenuation, and compatibility with the reduced channel model. Repeated calibration can reconstruct the Pauli-transfer matrix, but an endpoint random-unitary description alone does not establish slowly varying spatial drift: every unital qubit channel admits a random-unitary representation. Spatial regularity must be supported independently by segment-resolved characterization or by the response to an installed control sequence. Basis-dependent attenuation, leakage, or nonunital behavior requires a more general optical and security model.

Once the relevant spatial correlations are characterized, the natural design variable is the cell length $4\ell$, not merely the number of devices. A good starting point is to choose $4\ell$ smaller than the measured spatial correlation length of the polarization transformation but not so small that the insertion factor $t^{L/\ell}$ falls below the break-even threshold. The correlation length can be estimated indirectly by measuring the Jones matrix of successive fiber sections during installation, by using optical time-domain characterization in a test link, or by treating $Q_Z(\ell)$ and $Q_X(\ell)$ from a few pilot spacings as the operational proxy. The last option is often sufficient for QKD: the key-rate expression depends on the observed error rates and total transmission, not on the underlying stochastic parameters themselves.

Tab.~\ref{tab:workflow} summarizes a practical workflow. The table is intentionally formulated in terms of directly measurable quantities, so that the theorem supplies a conservative interpretation while the key-rate calculation uses empirical values. In particular, the deterministic bound is most useful before deployment, when one wants to choose a safe range of spacings from bounds on the local generator and its variation, or from a justified probabilistic certificate. After deployment, the measured quantities $(Q_Z,Q_X,\eta_{\rm ins})$ should replace the bound wherever possible.

\begin{table}[t]
\begin{adjustbox}{width=0.80\columnwidth}
\setlength{\tabcolsep}{15pt}
\setlength{\extrarowheight}{1pt}
\renewcommand{\arraystretch}{1.18}
\begin{tabular}{llll}
\toprule
stage & input & calculation & decision \\
\midrule
baseline & $Q_Z^0,Q_X^0,\eta_f$ & $K_0=q_{\rm sift}\eta_f r_0$ & establish endpoint baseline \\
loss test & device/cell transmission & $\eta_{\rm ins}(\ell)$ or $t^{L/\ell}$ & reject high-loss spacings \\
pilot scan & $Q_Z(\ell),Q_X(\ell)$ & $K(\ell)=q_{\rm sift}\eta_f\eta_{\rm ins} r(\ell)$ & keep $K(\ell)>K_0$ \\
stability & QBER time series & compare drift with/without cells & select stable optimum \\
operation & periodic test data & update $\alpha_x,\alpha_y,\alpha_z$ & bypass cells if loss dominates \\
\bottomrule
\end{tabular}
\end{adjustbox}
\caption{Operational workflow for applying passive Pauli toggling to a polarization-encoded QKD link.}
\label{tab:workflow}
\end{table}

This workflow also clarifies the role of the inserted Pauli rotations. They are not meant to replace all polarization control. Instead, they reduce the coherent part of the channel before the receiver performs basis measurements. If an active tracker is already present, reducing its required dynamic range or update bandwidth is a prospective benefit that must be tested with temporal data. If the receiver is passive or power-limited, the sequence can act as a purely optical preconditioner. In either case the same figure of merit applies: the sequence is useful only when the increase in the secret fraction is larger than the decrease in detection probability.

A second design issue is whether the Pauli frames should be placed uniformly. Uniform spacing is analytically clean and easy to manufacture, but the nonuniform form of Theorem~\ref{thm:end_to_end} suggests a better strategy for heterogeneous links. Sections of fiber under tight bends, thermal gradients, or mechanical stress can have larger local $h_m$ or larger intra-cell variation $\delta_m$. Such sections are candidates for shorter cells, subject to the additional loss and device-error budgets. Quiet sections can use longer cells; any uncontrolled sections require their own residual bounds. Mathematically this replaces $L/(4\ell)$ identical cells by a sum of local contributions,
\begin{eqnarray}
\eps_{\rm nonunif}=\sum_m\left(4\delta_m\ell_m+8h_m^2\ell_m^2\right),
\end{eqnarray}
with a separate measured insertion factor for each installed device. The optimum is then a constrained allocation problem: reduce the largest coherent-error contributions first, but stop adding devices when the additional loss outweighs the secret-fraction improvement in the product $T[1-h_2(Q_Z)-h_2(Q_X)]_+$.

It is useful to monitor three dimensionless diagnostics during such a scan. The first is the polarization-error suppression factor
\begin{eqnarray}
S_Q(\ell)=\frac{Q_Z^0+Q_X^0}{Q_Z(\ell)+Q_X(\ell)},
\end{eqnarray}
which measures only channel quality and ignores loss. The second is the insertion penalty
\begin{eqnarray}
G_\eta(\ell)=\frac{\eta_{\rm ins}(\ell)}{\eta_{\rm ins}(\infty)},
\end{eqnarray}
where $\eta_{\rm ins}(\infty)=1$ for an uncontrolled link. The third is the operational gain
\begin{eqnarray}
G_K(\ell)=\frac{K(\ell)}{K_0}.
\end{eqnarray}
These quantities need not be ordered. A design may have a very large $S_Q$ and still have $G_K<1$ if the inserted elements are too lossy. Conversely, a modest reduction of QBER can be enough for $G_K>1$ when the uncontrolled error rate is close to the entropy threshold and the insertion loss is small. Reporting all three numbers prevents the improvement from being attributed solely to a visually impressive QBER curve.

The same distinction matters for finite data. The channel and key-rate statements here are asymptotic, whereas an implemented QKD system estimates its errors from a finite sample. Confidence bounds $Q_Z^{\rm U},Q_X^{\rm U}$ can be propagated through the entropy terms on their increasing branch, but this substitution alone is not a complete finite-key security analysis. A composable finite-key rate must also include the applicable phase-error estimation, error-correction leakage, secrecy and correctness terms, and assumptions about temporal correlations. Likewise, Monte Carlo confidence intervals quantify simulation sampling uncertainty, not cryptographic security. The same physical insertion factor can be used with a separately justified finite-key formula under the stated scalar-loss assumptions.

%-------------------------------------------------------------------------------------------------------------------------------------------------------------------------------------------------------------------------------------
\section{Discussion and conclusion}
%-------------------------------------------------------------------------------------------------------------------------------------------------------------------------------------------------------------------------------------

We have developed a spatial-control framework for suppressing coherent polarization drift in optical fibers using a fixed sequence of passive in-line Pauli rotations. With propagation distance playing the role of time in a toggling frame, a two-segment echo exactly refocuses a constant generator that anticommutes with the inserted Pauli, while a four-frame cell visiting $\{\id,\hX,\hY,\hZ\}$ cancels the leading contribution of any traceless quasi-static generator. For fixed link length and fixed bounds on the strength and Lipschitz variation of the generator, we established an end-to-end residual-error bound that scales linearly with the segment spacing as $\ell\rightarrow0$. We then connected this coherent-control statement to operational QKD quantities: averaging over uncontrolled fiber realizations produces a unital random-unitary channel, and Pauli twirling preserves the diagonal transfer coefficients that determine the BB84 error rates. In the representative 50 km stochastic-fiber simulation, toggling at $\ell=1.25$ km reduced $(Q_Z,Q_X)$ from approximately $(0.060,0.062)$ to $(2.54\times10^{-3},1.93\times10^{-3})$. With an assumed per-device transmission of $t=0.997$, the corresponding asymptotic key rate increased from $1.68\times10^{-2}$ to $4.23\times10^{-2}$ bits per launched pulse, or by approximately a factor of $2.5$.

The central lesson in our analysis is that the inserted devices repeatedly change the frame in which a slowly varying rotation accumulates, preventing its leading contribution from adding coherently. This explains both the appeal and the boundary of the method. The deterministic theorem provides a worst-case residual-error bound for the modeled unitary component when the generator is bounded and varies slowly within each cell, whereas the Monte Carlo results describe the representative behavior of the particular Ornstein--Uhlenbeck ensemble and parameter set considered here. The Pauli-twirled description likewise should not be interpreted as an assumption that the physical fiber channel is intrinsically Pauli diagonal: it is an operational reduction that preserves the BB84 QBERs relevant to the present analysis. Polarization-dependent loss, polarization-mode dispersion, and spectral correlations cannot be removed by this unitary symmetrization; source and detector imperfections and miscalibration of the embedded rotations require separate modeling. In our simulations, insertion loss eventually outweighs further QBER suppression; accumulated device error would introduce an additional spacing-dependent cost. The physically meaningful result is therefore the possibility of a finite design window, when the break-even condition can be met, determined jointly by the spatial correlation of the birefringence, the device spacing, and the measured transmission, rather than a universal prescription for the number of devices.

Within this interpretation, spatial Pauli toggling offers a practical way to shift part of polarization stabilization from continuous feedback to the structure of the transmission channel itself. The scheme could serve as a passive optical preconditioner for low-power or difficult-to-access links. In a hybrid architecture, it might also reduce the dynamic range and update rate required of an active polarization tracker. The nonuniform form of our bound further suggests placing shorter cells only in fiber sections with strong bending, thermal gradients, or mechanical stress, while leaving quieter sections sparsely controlled. Important next steps are to test the protocol with measured Jones-matrix data, incorporate lossy and misaligned devices directly into the channel model, and combine the observed QBER improvement with finite-key estimation and realistic source and detector statistics. Extensions to entanglement-based, decoy-state, and measurement-device-independent protocols can then determine how much of the channel-level gain survives in complete communication systems. More broadly, the present results illustrate a useful design principle: a quantum channel need not remain merely a source of errors to be corrected at its endpoints; when its coherent dynamics possess sufficient spatial structure, the channel itself can become part of the control architecture.

%==================================================================================================================================
\section*{Acknowledgement}
%==================================================================================================================================

This work was supported by the Ministry of Trade, Industry and Resources (MOTIR), Korea, under the project ``Industrial Technology Infrastructure Program'' (RS-2024-00466693). This work is also supported by the Grant No.~K25L5M2C2 at the Korea Institute of Science and Technology Information (KISTI). We acknowledge the Yonsei University Quantum Computing Project Group for providing support and access to the Quantum System One (Eagle Processor), which is operated at Yonsei University.

%====================================================================================================================================
\appendix
%====================================================================================================================================

%-------------------------------------------------------------------------------------------------------------------------------------------------------------------------------------------------------------------------------------
\section{Pauli algebra and Bloch representation}
%-------------------------------------------------------------------------------------------------------------------------------------------------------------------------------------------------------------------------------------

The Pauli matrices satisfy
\begin{eqnarray}
\hsigma_i\hsigma_j=\delta_{ij}\id+i\sum_k\eps_{ijk}\hsigma_k,
\end{eqnarray}
which implies the vector identity
\begin{eqnarray}
(\bm a\cdot\sig)(\bm b\cdot\sig)=(\bm a\cdot\bm b)\id+i(\bm a\times\bm b)\cdot\sig.
\label{eq:pauli_vector_identity_app}
\end{eqnarray}
Every qubit state can be written uniquely as
\begin{eqnarray}
\hrho=\frac{1}{2}(\id+\bm r\cdot\sig), \quad \norm{\bm r}\le1.
\end{eqnarray}
The eigenvalues are $(1\pm\norm{\bm r})/2$, so positivity is equivalent to $\norm{\bm r}\le1$. For $\hU\in SU(2)$, define
\begin{eqnarray}
(R_{\hU})_{ij}=\frac{1}{2}\tr{\hsigma_i\hU\hsigma_j\hU^\dagger}.
\end{eqnarray}
Then, $R_{\hU}\in SO(3)$ and
\begin{eqnarray}
\hU(\bm a\cdot\sig)\hU^\dagger=(R_{\hU}\bm a)\cdot\sig.
\end{eqnarray}
Eq.~\eqref{eq:bloch_transport} follows by substituting $\hrho=(\id+\bm r\cdot\sig)/2$ and $\hH=(\be\cdot\sig)/2$ into $d\hrho/dz=-i[\hH,\hrho]$, then using Eq.~\eqref{eq:pauli_vector_identity_app}.

For a unit vector $\bm m$, the Hermitian unitary $\hP_{\bm m}=\bm m\cdot\sig$ acts by
\begin{eqnarray}
\hP_{\bm m}(\bm n\cdot\sig)\hP_{\bm m}=(2(\bm m\cdot\bm n)\bm m-\bm n)\cdot\sig.
\end{eqnarray}
Thus, if $\bm m\cdot\bm n=0$, conjugation by $\hP_{\bm m}$ flips the sign of $\bm n\cdot\sig$. This is the algebraic basis of Theorem~\ref{thm:echo}.

%-------------------------------------------------------------------------------------------------------------------------------------------------------------------------------------------------------------------------------------
\section{Dyson and telescoping estimates}
%-------------------------------------------------------------------------------------------------------------------------------------------------------------------------------------------------------------------------------------

Let $\hHtf(z)$ be a bounded Hermitian generator on $[0,T]$ and let
\begin{eqnarray}
\hU(T,0) = \mathcal{P}\exp\left[-i\int_0^T\hHtf(z)dz\right].
\end{eqnarray}
The Dyson expansion gives
\begin{eqnarray}
\hU(T,0) = \id+\sum_{n=1}^{\infty}(-i)^n\int_{0\le z_n\le\cdots\le z_1\le T}\hHtf(z_1)\cdots\hHtf(z_n)dz_1\cdots dz_n.
\end{eqnarray}
Writing $\hM_1=\int_0^T\hHtf(z)dz$ and $\Lambda=\int_0^T\norm{\hHtf(z)}dz$, absolute summation of the Dyson terms gives the standard estimate $\norm{\hU(T,0)-\id}\le\norm{\hM_1}+e^\Lambda-1-\Lambda$. For a Hermitian generator, unitarity gives the sharper exact remainder estimate used in Theorem~\ref{thm:end_to_end}. Inserting the integral equation for $\hU(s,0)$ once yields
\begin{eqnarray}
\hU(T,0)-\id+i\hM_1 = -\int_0^Tds\int_0^sdr\,\hHtf(s)\hHtf(r)\hU(r,0).
\end{eqnarray}
Since $\norm{\hU(r,0)}=1$,
\begin{eqnarray}
\norm{\hU(T,0)-\id} \le \min\left\{2,\norm{\hM_1}+\frac{\Lambda^2}{2}\right\} \le \norm{\hM_1}+e^\Lambda-1-\Lambda.
\label{eq:dyson_bound_app}
\end{eqnarray}
This bound retains all orders of the exact evolution through its unitary remainder, rather than discarding higher Dyson terms.

The product estimate used there is also elementary. If $\hVop_1,\ldots,\hVop_M$ are unitary and $\hW_k=\hVop_k\cdots \hVop_1$, then
\begin{eqnarray}
\norm{\hW_k-\id}=\norm{\hVop_k(\hW_{k-1}-\id)+(\hVop_k-\id)} \le \norm{\hW_{k-1}-\id}+\norm{\hVop_k-\id}.
\end{eqnarray}
Iteration yields $\norm{\hVop_M\cdots \hVop_1-\id}\le\sum_m\norm{\hVop_m-\id}$.

%-------------------------------------------------------------------------------------------------------------------------------------------------------------------------------------------------------------------------------------
\section{Pauli twirling and BB84 formulas}
%-------------------------------------------------------------------------------------------------------------------------------------------------------------------------------------------------------------------------------------

For a unital qubit channel, $\calE(\hrho)=(\id+(A\bm r)\cdot\sig)/2$. If Alice sends a $\hZ$-basis bit, the average probability that Bob obtains the wrong $\hZ$ outcome is
\begin{eqnarray}
Q_Z(\calE)=\frac{1}{2}\left(1-\frac{1}{2}\tr{\hZ\calE(\hZ)}\right)=\frac{1-\alpha_z}{2}.
\end{eqnarray}
The $\hX$-basis expression is identical with $\hZ$ replaced by $\hX$. A Pauli channel acts diagonally on Pauli operators,
\begin{eqnarray}
\calP(\hX)=(p_I+p_X-p_Y-p_Z)\hX,
\quad
\calP(\hZ)=(p_I-p_X-p_Y+p_Z)\hZ,
\end{eqnarray}
so $Q_Z=p_X+p_Y$ and $Q_X=p_Z+p_Y$.

The one-way asymptotic BB84 secret fraction in the single-photon, i.i.d. Pauli-channel setting is bounded by
\begin{eqnarray}
r_{\rm BB84}\ge\left[1-h_2(e_b)-h_2(e_p)\right]_+,
\end{eqnarray}
where $e_b=Q_Z$ and $e_p=Q_X$. The same expression follows from the entanglement-based picture: Pauli noise on Bob's half of $\ket{\Phi^+}$ produces a Bell-diagonal state~\cite{Bennett1996EntanglementPurification} with weights $(p_I,p_X,p_Y,p_Z)$, and the full-weight hashing expression $[1-H_4(p_I,p_X,p_Y,p_Z)]_+$ is lower-bounded by the two-error expression through subadditivity of Shannon entropy. Access to the full Bell weights requires information beyond the two BB84 marginals. A certified error upper bound $q$ gives $r_{\rm BB84}\ge[1-2h_2(\min\{q,1/2\})]_+$; the truncation at $1/2$ is required by the monotonicity range of binary entropy.

%-------------------------------------------------------------------------------------------------------------------------------------------------------------------------------------------------------------------------------------
\section{Simulation algorithm}
%-------------------------------------------------------------------------------------------------------------------------------------------------------------------------------------------------------------------------------------

For each realization of Eq.~\eqref{eq:ou_model}, the code computes fine-step propagators
\begin{eqnarray}
\hU_k=\exp\left[-\frac{i\Delta z}{2}\be(z_k)\cdot\sig\right]
\end{eqnarray}
and cumulative products $\hC_m=\hU_{m-1}\cdots \hU_0$. A coarse segment propagator from fine index $a$ to $b$ is obtained as $\hC_b\hC_a^\dagger$, avoiding repeated multiplication for each spacing. For $N$ controlled segments, the code forms
\begin{eqnarray}
\hVop_{\rm ctrl}(N)=\prod_{j=N}^{1}\hg_{j-1}^\dagger \hU(z_j,z_{j-1})\hg_{j-1},
\quad
\hg_{j-1}\in\{\id,\hX,\hY,\hZ\}
\end{eqnarray}
with cyclic order $\id,\hX,\hY,\hZ$. This is the phase-fixed residual; its channel statistics equal those of the cyclic physical propagator. For each realization, it computes $R_{\hU}$, $Q_Z=(1-R_{zz})/2$, $Q_X=(1-R_{xx})/2$, and $F_e=|\tr{\hU}|^2/4$. The reported channel quantities are ensemble means; the plotted error bars are $1.96$ times the standard error of the mean.

\bibliographystyle{apsrev4-2}
\bibliography{botherq_refs}

\end{document}